\documentclass[preprint,12pt]{elsarticle}
\usepackage{lineno,hyperref}
\usepackage{amsmath,amssymb}
\usepackage{amsthm}
\usepackage{graphicx}
\usepackage{float}
\usepackage{subeqnarray}
\usepackage{cases}
\usepackage{setspace}

\newtheorem{theoremx}{Theorem}
\newtheorem{corollaryx}{Corollary}
\newtheorem{remark}{Remark}
\newtheorem{lemma}{Lemma}

\journal{Journal of The Franklin Institute, 2016, DOI: 10.1016/j.jfranklin.2016.06.025}

\begin{document}

\begin{frontmatter}

\title{Task-space coordinated tracking of multiple heterogeneous manipulators via \\ controller-estimator approaches}

\author{Ming-Feng~Ge$^a$}
\ead{fmgabc@163.com}
\author{Zhi-Hong~Guan$^{a}$}
\ead{zhguan@mail.hust.edu.cn}
\author{Chao Yang$^a$}
\author{Chao-Yang Chen$^b$}
\author{Ding-Fu Zheng$^a$}
\author{Ming Chi$^a$}

\address{$^a$College of Automation, Huazhong University of Science and Technology,\\
Wuhan, 430074, China \\
$^b$School of Information and Electrical Engineering, Hunan University of Sience and Technology
Xiangtan, Hunan, 411201, China}

\begin{abstract}
This paper studies the task-space coordinated tracking of a time-varying leader for multiple heterogeneous manipulators (MHMs),
containing redundant manipulators and nonredundant ones.
Different from the traditional coordinated control, distributed controller-estimator algorithms (DCEA),
which consist of local algorithms and networked algorithms, are developed
for MHMs with parametric uncertainties and input disturbances.
By invoking differential inclusions, nonsmooth analysis, and input-to-state stability,
some conditions (including sufficient conditions, necessary and sufficient conditions)
on the asymptotic stability of the task-space tracking errors and the subtask errors are developed.
Simulation results are given to show the effectiveness of the presented DCEA.
\end{abstract}



\end{frontmatter}



\section{Introduction}
Coordinated control has been diffusely invoked in many practical applications of multiple manipulators,
including coordination of bilateral human-swarm systems {\cite{fmg05}}, single-master-multiple-slaves teleoperation {\cite{LiSu}}, dual-user shared teleoperation {\cite{YHG}}-{\cite{LHLM03}},
multi-robot teleoperation {\cite{Liu02}}-{\cite{LDGD05}},
multi-fingered grasping and manipulation {\cite{fmg03,Gueaieb01}},
due to their prominent superiority comparing with traditional centralized control,
such as stronger stability, less energy consumption, greater operational efficiency {\cite{Olfati}}-{\cite{Guan04}}.

Existing works focused on joint-space synchronization of two manipulators, namely, master-slave bilateral teleoperators, using various control technologies, such as adaptive control, passivity-based control, proportional-derivative control {\cite{Imaida}}-{\cite{Islam}}.
However, the single-master-single-slave framework containing two manipulators is already too simple for some complicated practical tasks.
By invoking cooperative concurrent control, synchronization of interconnected Lagrangian systems had been investigated {\cite{Chung}}.
Distributed containment control had been developed for nonlinear multi-agent systems
in the presence of dynamical uncertainties and external disturbances and applied to Lagrangian networks {\cite{Mei01}}.
However, the kinematics of robotic manipulators had not been taken into consideration in the above literatures.
In reality, task-space algorithms considering kinematics of manipulators are more practical and applicative comparing with joint-space algorithms.
It thus motivates a group of researches on task-space algorithms.
In presence of kinematic and dynamic uncertainties, task-space synchronization had been addressed for multiple manipulators
under strong connected graphs by invoking passivity control {\cite{Wang01}} and adaptive control {\cite{Wang02}}.
Note that the above literatures mainly focuses on motion control of kinematically identical and nonredundant robotic manipulators
with parametric uncertainties.

Redundant manipulators can achieve more performance benefits in contrast to nonredundant ones {\cite{Zergeroglu01,XuZhangL}}.
On the other hand, multiple manipulators containing both redundant and nonredundant individuals, namely,
multiple heterogeneous manipulators (MHMs), are necessary and inevitable in some natural and man-made systems,
$e.g.$, human hands and multi-fingered hands {\cite{fmg03,Gueaieb01}}.
Inspired by conceivable performance benefits of MHMs, task-space synchronization of MHMs
under balanced connected topologies had been achieved via the passivity control technology in {\cite{Liu01}},
which can synchronize the combination signals of task-space position/velocity tracking errors to zero.

In this paper, we present some novel DCEA to achieve the task-space coordinated tracking of a time-varying leader
for MHMs with parametric uncertainties and input disturbances, and the interaction topology of the MHMs are
assumed to contain a spanning tree.
The main contributions are summarised as follows.
1) Different from the coordination algorithms considering dynamics of Euler-Lagrange systems {\cite{Chung,Mei01}},
both the dynamics and kinematics are considered, which is more practical and challenging.
2) Different from the coordination algorithms for identical manipulators with
a constant agreement value {\cite{Wang01,Wang02}}, the tracking of a time-varying leader for MHMs is studied.
3) Different from balanced interaction topologies studied in {\cite{Liu01}},
digraphs containing a spanning tree are invoked to describe the interaction topology.
4) The novel controller-estimator structure provides a theoretical guidance for coordinated control
of various networked multi-agent systems, whose dynamics are complex and strong nonlinear.

The rest of the paper is structured as follows.
In Section 2, the preliminaries are presented.
In Section 3, the main results are presented.
In Section 4, the simulations are given.
In Section 5, the conclusion is proposed.

\emph{Notation:} ${\mathbb R} ^ n$ represents the $n$-dimensional Euclidean space, $I_n$ denotes the $n \times n$ identity matrix,
$1_n = [1,\cdots,1]^T$ represents the $n$-dimensional column vector, $\left\| \cdot \right\|$ represents the Euclidean norm,
$\lambda_{\min}(\cdot)$ denotes the minimum eigenvalue of the corresponding matrix.

\section{Preliminaries}

\subsection{ Dynamics and Kinematics}

The dynamics and kinematics of the $i$th individual in the MHMs
with input disturbances are given as follows {\cite{Zergeroglu01}}:
\begin{equation}\label{1.1}
\left\{ \begin{array}{lll}
{H_i}({q_i}){{\ddot q}_i} + {C_i}({q_i},{{\dot q}_i}){{\dot q}_i}
+ {g_i}({q_i}) + {d_i}(t) = {u_i},\\ \\
{x_i} = {\psi _i}({q_i}),~{{\dot x}_i} = {J_i}({q_i}){{\dot q}_i},
\end{array} \right.
\end{equation}
where $i \in {\mathcal V} = \{ 1,\ldots,n \}$,
$t \in \mathcal{Q} = [t_0,\infty)$, $t_0 \geq 0$ is the initial time,
${q_i} \in {\mathbb R}^{p_i}$ denotes the position in the joint space,
$p_i \geq 2$ represents the degrees-of-freedom (DOF) of the $i$th manipulator,
${ H_i}(q_i) \in {\mathbb R}^{p_i \times p_i}$ is the inertia matrix,
${C_i}({q_i},{{\dot q}_i}) \in {\mathbb R}^{p_i \times p_i}$ represents the centripetal-Coriolis matrix,
${g_i}(q_i) \in {\mathbb R}^{p_i}$ denotes the gravity vector,
$d_i(t) \in {\mathbb R}^{p_i}$ represents the input disturbance,
${u_i} \in {\mathbb R}^{p_i}$ stands for the torque input,
$x_i \in \Omega_x^i \subseteq {\mathbb R}^{m}$ denotes the task-space position,
$\Omega_x^i$ represents the work space,
${\psi _i}(q_i) \in {\mathbb R}^{m}$ denotes the forward kinematics,
${m} \geq 2$ denotes the task-space dimension and is thus the minimum number of DOF required to
perform a given end-effector task,
${J_i}(q_i) = \partial {\psi _i}({q_i})/ \partial {q_i} \in {{\mathbb R}^{{m} \times p_i}}$
stands for the Jacobian matrix.

The MHMs contain redundant manipulators and nonredundant ones, which means ${\mathcal V}$ consists of two subsets
$\mathcal{E} = \left\{ {i \in \mathcal{V}~|~{p_i} = {m}} \right\}$ and
$\mathcal{F} = \left\{ {i \in \mathcal{V}~|~{p_i} > {m}} \right\}$, $i.e.$,
the manipulators in $\mathcal{E}$ are nonredundant and the ones in $\mathcal{F}$ are redundant.
It is worthy to point out that a nonredundant manipulator has just a necessary number of DOF
to perform a given end-effector task ($i.e.$, maintask)
and a redundant manipulator has more DOF than the necessary number,
which results in an infinite number of joint configurations to the inverse-kinematics problem.
The redundancy can improve the functionality and flexibility of robotic manipulators and many subtasks
(including manipulability enhancement, mechanical limit avoidance,
and obstacle avoidance) can be obtained by choosing appropriate joint configurations.
Then the following algebraic operation, that will be invoked hereinafter, is given by
\begin{equation*}
  J_i^\sharp =
  \left\{ \begin{array}{lll}
    J_i^{-1}, & i \in \mathcal{E}, \\ \\
    J_i^\dag, & i \in \mathcal{F},
  \end{array}
  \right.
\end{equation*}
where for $i \in \mathcal{E}$, $J_i^\sharp = J_i^{-1} \in {\mathbb R}^{m \times m}$ denotes the normal inverse of $J_i (q_i)$;
for $i \in \mathcal{F}$, $J_i^\sharp = J_i^\dag = J_i^T{({J_i}J_i^T)^{ - 1}} \in {\mathbb R}^{p_i \times m}$
represents the pseudoinverse of $J_i (q_i)$ and satisfies the well-known Moore-Penrose conditions \cite{Nakamura}.
The following lemma that will be used hereinafter is given.

\begin{lemma}\label{l1}
For $i \in \mathcal{F}$, the algebraic operation $J_i^\sharp$ satisfies
\begin{equation*}
\begin{array}{lll}
{J_i}({I_{p_i}} - J_i^ \sharp {J_i}) = 0,~({I_{p_i}} - J_i^ \sharp {J_i})J_i^ \sharp  = 0,\\ \\
({I_{p_i}} - J_i^ \sharp {J_i})({I_{p_i}} - J_i^ \sharp {J_i}) = {I_{p_i}} - J_i^ \sharp {J_i}.
\end{array}
\end{equation*}
\end{lemma}

For any $i \in \mathcal V$, the properties of system (\ref{1.1}) are given as follows \cite{Liu01,Kelly}. \\
(P1) ${H_i}({q_i})$ is positive definite. ${\dot H_i}({q_i}) - 2 C_i(q_i,\dot q_i)$ is skew symmetric; \\
(P2) The dynamic terms ${H_i}({q_i})$, ${g_i}({q_i})$, $d_i(t)$ are bounded for all possible $q_i$,
and $\| {C_i}({q_i},{{\dot q}_i}) \| \leq {\bar c}_i \| {\dot q}_i \|$, where ${\bar c}_i > 0$ is a positive constant;\\
(P3) The dynamic terms can be parameterized, $i.e.$, ${H_i}({q_i})x + {C_i}({q_i},{{\dot q}_i})y + {g_i}({q_i}) = Y_i (q_i, \dot q_i, y, x) \vartheta_i$, where $x,y$ are any proper vectors, $Y_i (q_i, \dot q_i, y, x)$ is the regressor and $\vartheta_i$ is a set of constant dynamic parameters.

\begin{remark}
By the actual characteristics of MHMs, we assume that the kinematic terms ${J_i}({q_i})$ and $J_i^ \sharp(q_i)$ are bounded;
the kinematic singularities are avoided, $i.e.$, ${\rm rank} (J_i^ \sharp(q_i)) = m$, for $i \in \mathcal V$.
The above assumptions are the general properties of multiple manipulators, see \cite{Zergeroglu01,Kelly} for details.
\end{remark}

\subsection{ Graph Theory}

The interaction of the MHMs is denoted by a digraph $\mathcal{G} = \left\{ \mathcal{V} ,\xi,\mathcal{A} \right\}$,
where ${\mathcal V}$ is the node set, ${\xi} \subseteq \mathcal{V} \times \mathcal{V}$ is the edge set,
$\mathcal{A} = {\left[ { \varepsilon_{ij}} \right]_{n \times n}}$ is the weighted adjacency matrix with nonnegative adjacency elements.
An edge in $\mathcal{G}$ is denoted by an ordered pair $(v_i,v_j)$.
$(v_i,v_j) \in \xi$ if and only if node $j$ ($i.e.$, the $j$th manipulator) can directly access the information of node $i$.
${{\mathcal N}_j} = \left\{ {i \in {\cal V}~|~(v_i,v_j) \in \xi } \right\}$ denotes the neighbor set of node $j$.
$\mathcal{A}$ is defined as $i \in {{\cal N}_j}  \Leftrightarrow {\varepsilon_{ji}} > 0$, otherwise ${\varepsilon_{ji}} = 0$, and ${\varepsilon_{ii}} = 0$, $\forall i,j \in \mathcal{V}$.
$\mathcal{D} = diag\{ d_1, \ldots , d_n \} $ is the degree matrix,
where ${d_i} = \sum\nolimits_{j \in {\mathcal N}_i} {\varepsilon _{ij}}$.
The Laplacian matrix is defined as $\mathcal{L}=\mathcal{D}-\mathcal{A}$.
A directed path is an ordered sequence $v_1,v_2,\cdots,v_\omega$ satisfying that any ordered pair of vertices appearing consecutively in the sequence is an edge of the digraph $\mathcal{G}$.
A digraph contains a spanning tree if there exists a root node that has a directed path to all the other nodes.
$\mathcal{B} = [b_1,\ldots,b_n]^T$ is the weight vector between the $n$ nodes and the leader,
where $b_i > 0$ if the states of the leader is available to node $i$, namely, node $i$ is pinned; $b_i = 0$ otherwise.
The states of the leader (node $0$) $x_0$, $v_0$ and $a_0 \in {\mathbb R}^{m}$ satisfy $\dot x_0 = v_0$, $\dot v_0 = a_0$.
The following assumptions that will be used hereinafter are given. \\
(A1) The leader has a directed path to all the nodes in $\mathcal G$ under $\mathcal B$; \\
(A2) $\| v_0 \|_{\infty} < \beta_1$, $\| a_0 \|_{\infty} < \beta_2$ and $\| \dot a_0 \|_{\infty} < \beta_3$, where $\beta_1$, $\beta_2$ and $\beta_3$ are positive constants.

\begin{remark}
By Assumption {A2}, the derivatives of the states of the leader $x_0$, $v_0$ and $a_0$
are bounded, which happens to be the actual characteristics of the trajectories that can be
reachable by the real-world manipulators described by Euler-Lagrange equations {\cite{MengDi}}.
\end{remark}

\subsection{Problem Statement}

The control tasks of this paper are presented in this section.
It can be seen that for given $\dot x_i$, ${{\dot x}_i} = {J_i}({q_i}){{\dot q}_i}$ admits an infinite number of $\dot q_i$ when $i \in \mathcal{F}$
and only a single solution when $i \in \mathcal{E}$.
Thus, to accomplish a given end-effector task, there exist only a single joint configuration for nonredundant manipulators
and an infinite number of joint configurations for redundant manipulators.
Let the given end-effector task be called maintask and the tracking of a class of special joint configuration be called subtasks.
Then the control tasks in this paper can be divided into the maintask and the subtask.
Moreover, note that the nonredundant manipulators which have only one joint configuration
for a given end-effector task ($i.e.$, maintask) cannot accomplish the subtasks.
\par
{Maintask}:
The maintask is to design distributed input $u_i$ with the states of node $i$ (the $i$th manipulator)
and its neighbour set such that the task-space coordinated tracking can be accomplished, $i.e.$,
${x_i} \to {x_0}$ and ${\dot x_i} \to {v_0}$ as $t \to \infty$, $\forall i \in {\mathcal V}$.
\par
{Subtask}:
The subtask, $e.g.$, manipulability enhancement, mechanical limit avoidance,
and obstacle avoidance, can only be accomplished by redundant manipulators.
For different subtasks with respect to different redundant manipulators,
we can construct a corresponding auxiliary vector $\varphi _i(t)$
to denote the gradients of some performance indices with respect to the corresponding subtask.
Then the subtask for the $i$th manipulator is said to be achieved if $e_{si} \to 0$ as $t \to \infty$,
where $i \in {\mathcal F}$, ${e_{si}} = ( I_{{p_i}} - J_i^ \sharp {J_i} )( \dot q_i - {\varphi _i} )$
denotes the subtask tracking error {\cite{Hsu}}.

\begin{remark}
The subtasks give more functional constraints on joint configurations and thus cannot be accomplished by nonredundant manipulators
because they have only one joint configuration for their inverse-kinematics.
Redundant manipulators have an additional DOF to accomplish the subtasks, which gives higher robustness and wider operational space
comparing with nonredundant ones. However, the control design for redundant manipulators is more complex and challenging.
Human arms, elephant trunks, and snakes are some examples of this kind of redundant system.
The redundant is applicable in many practical applications and tough challenging in theoretical analysis \cite{fmg03,Gueaieb01,Wang01}.
\end{remark}

\section{Task-space Coordinated Tracking of Multiple Heterogeneous Manipulators}
This section studies the task-space coordinated tracking for MHMs with a time-varying leader,
in which the position vector of the leader in the generalized coordinate is assumed to be bounded up to its third derivative,
which is also invoked in {\cite{MengDi}}.

\subsection{Distributed Controller-Estimator Algorithms}

In this section, DCEA for task-space coordinated tracking of MHMs are developed.
Let $\hat{x}_{i}$, ${\hat{v}_{i}}$ and ${\hat{a}_{i}} \in {\mathbb R}^{m}$ be, respectively,
the estimated value of $x_0$, $v_0$ and $a_0$ for the $i$th manipulator.
Considering the heterogeneity of the MHMs, a joint-space auxiliary velocity
${\dot{\hat q}_{ri}} \in {\mathbb R}^{p_i}$ is given by
\begin{equation}\label{1.5}
  {\dot{\hat q}_{ri}} = J_i^\sharp ({\hat v}_i - {\alpha _i} (x_i - {\hat x}_i)) + (I_{p_i} - J_i^ \sharp {J_i}){\varphi _i},
\end{equation}
where $\alpha _i$ is a positive constant,
$\varphi _i \in {\mathbb R}^{p_i}$ is the gradients of some performance indices with respect to the redundant manipulators.
Then the joint-space auxiliary acceleration ${\ddot{\hat q}_{ri}}$ is defined as
\begin{equation}\label{1.6}
  {\ddot{\hat q}_{ri}} =
  \dot J_i^\sharp ({\hat v}_i - {\alpha _i} (x_i - {\hat x}_i))
  + J_i^\sharp ( {{\hat a}_i} - {\alpha _i} (\dot x_i - {\hat v}_i) )
  + \frac{d}{dt}[ ( {I_{p_i}} - J_i^\sharp {J_i}){\varphi _i} ].
\end{equation}
Let $\zeta_i = {\rm col} ({\hat x}_i,{\hat v}_i,{\hat a}_i) \in {\mathbb R}^{3m}$.
Then for $i,j \in \mathcal V$, we define
\begin{equation}\label{0.9}
  \sigma_{ij} = \zeta_i - \zeta_j,
\end{equation}
especially, $\sigma_i = \zeta_i - {\rm col} (x_0,v_0,a_0)$.

Let ${\hat s_i} = {\dot q}_i - {\dot{\hat q}_{ri}}$.
In the presence of uncertain dynamics and input disturbance, the DCEA for MHMs consist of the control law
\begin{equation}\label{1.2}
{u_i} =  Y_i( {q_i},{\dot q}_i,{\dot{\hat q}_{ri}},{\ddot{\hat q}_{ri}} ) {\hat \vartheta}_i - J_i^T {\mathcal{K}_{xi}}J_i {{{\hat s }_i}} - {\mathcal{K}_{si}}  {{{\hat s }_i}} - {\mathcal{K}_{ri}} {\rm sgn} ( \hat s_i ),
\end{equation}
and the distributed adaptive estimators
\begin{subequations}\label{1.7}
\begin{numcases}{}
   {{\dot {\hat \vartheta} }_i} = - {\mathcal T}_i Y_i^T( {q_i},{\dot q}_i,{\dot{\hat q}_{ri}},{\ddot{\hat q}_{ri}} ){{\hat s}_i}, \label{1.8} \\ \nonumber\\
   {\dot \zeta_i} = - \left( \left[ {\begin{array}{*{20}{lll}}
                            {{\beta _1}}&{}&{}\\
                            {}&{{\beta _2}}&{}\\
                            {}&{}&{{\beta _3}}
                            \end{array}} \right] \otimes I_{m} \right)
    {\rm sgn} \left( {\sum\limits_{j \in {\mathcal N}_i} {{\varepsilon_{ij}}\sigma_{ij}} } + b_i \sigma_i \right),\label{1.9}
\end{numcases}
\end{subequations}
where $\beta_1$, $\beta_2$ and $\beta_3$ are given in Assumption A2,
$\otimes$ denotes the Kronecker product, $\sigma_{ij}$ is given in (\ref{0.9}),
${{\hat \vartheta }_i}$ is the estimated value of ${ \vartheta _i}$,
$ Y_i( {q_i},{\dot q}_i,{\dot{\hat q}_{ri}},{\ddot{\hat q}_{ri}} ){{\hat \vartheta }_i} = {\hat H_i}({q_i}){\ddot{\hat q}_{ri}} + {\hat C_i}({q_i},{{\dot q}_i}){\dot{\hat q}_{ri}} + {\hat g_i}({q_i})$,
${\hat{H}_i}({q_i})$, ${\hat{C}_i}({q_i},{{\dot q}_i})$, and ${\hat{g}_i}({q_i})$ are the estimates
of ${H_i}({q_i})$, ${C_i}({q_i},{{\dot q}_i})$, and ${g_i}({q_i})$ respectively,
$\mathcal{K}_{si}$, $\mathcal{K}_{ri} \in {\mathbb R}^{p_i \times p_i}$ and ${\mathcal{K}_{xi}} \in {\mathbb R}^{m \times m}$ are positive definite matrices, ${\mathcal T}_i$ is a designed diagonal positive definite matrix with an appropriate dimension.

\begin{remark}
It can be seen from the definition of $J_i^\sharp$ that
$I_{p_i} - J_i^ \sharp {J_i} = 0$ for $i \in \mathcal E$.
Therefore, (\ref{1.5}) implies that for redundant manipulators and nonredundant ones,
${\dot{\hat q}_{ri}}$ is distinguishing, namely,
\begin{equation*}
  {\dot{\hat q}_{ri}} = \left\{ \begin{array}{lll}
  J_i^{-1} ({\hat v}_i - {\alpha _i} (x_i - {\hat x}_i)),& i \in {\mathcal E}, \\ \\
  J_i^\dag ({\hat v}_i - {\alpha _i} (x_i - {\hat x}_i)) + (I_{p_i} - J_i^ \dag {J_i}){\varphi _i},& i \in {\mathcal F}.
  \end{array} \right.
\end{equation*}
Besides, ${\ddot{\hat q}_{ri}}$ is distinguishing with respect to redundant manipulators and nonredundant ones.
It follows that the control law (\ref{1.2}) is distinguishingly designed with respect to different sets of manipulators.
\end{remark}

\subsection{Boundedness Analysis}

In this section, by analyzing the boundedness of the system states,
the simplification of the close-loop dynamics is given.

The normal variables ${\dot q_{ri}}$, ${\ddot q_{ri}}$ and ${s_i}$ with respect to
the auxiliary variables ${\dot{\hat q}_{ri}}$, ${\ddot{\hat q}_{ri}}$ and
${\hat{s} _i}$, containing estimated states, are defined as
\begin{equation}\label{1.3}
\begin{array}{lll}
  {\dot q_{ri}} =
  J_i^\sharp (v_0 - {\alpha _i} (x_i - x_0)) + (I_{p_i} - J_i^\sharp {J_i}){\varphi _i},\\ \\
   {\ddot q_{ri}} =
  \dot J_i^\sharp (v_0 - {\alpha _i} (x_i - x_0))
  + J_i^\sharp ( a_0 - {\alpha _i} (\dot x_i - v_0) )
  + \frac{d}{dt}[ ( {I_{p_i}} - J_i^\sharp {J_i}){\varphi _i} ], \\ \\
  {s _i} = {\dot q}_i - {\dot q_{ri}},
\end{array}
\end{equation}
where the normal variables are formed based on the estimated variables
by replacing $\hat x_i$, $\hat v_i$ and $\hat a_i$ with $x_0$, $v_0$ and $a_0$ respectively. Then we define
\begin{equation}\label{1.0}
\begin{array}{lll}
  \dot{\tilde q}_{ri} &=& \dot{\hat q}_{ri} - \dot q_{ri} \\ \\
  &=& J_i^\sharp (\hat v_i - v_0 + {\alpha _i} ( \hat x_i - x_0)),
  \end{array}
\end{equation}
and
\begin{equation}\label{2.2}
\begin{array}{lll}
  \ddot{\tilde q}_{ri} &=& \ddot{\hat q}_{ri} - \ddot q_{ri} \\ \\
  &=& \dot J_i^\sharp (\hat v_i - v_0 + {\alpha _i} ( \hat x_i - x_0)) + J_i^\sharp (\hat a_i - a_0 + {\alpha _i} ( \hat v_i - v_0)).
  \end{array}
\end{equation}

Let $\tilde s_i = \hat s_i - s_i = - \dot{\tilde q}_{ri}$.
Substituting (\ref{1.2}) into (\ref{1.1}) yields the following close-loop dynamics
\begin{equation}\label{2.3}
\begin{array}{lll}
  {H_i}({q_i}){\dot s_i} + {C_i}({q_i},{{\dot q}_i}){s_i} + {d_i}(t)
  + J_i^T {\mathcal{K}_{xi}}J_i {s_i} + {\mathcal{K}_{si}} {s_i} + {\mathcal{K}_{ri}} {\rm sgn} ( s_i ) \\ \\
  ~~~~~= Y_i( {q_i},{\dot q}_i,{\dot q_{ri}},{\ddot q_{ri}} ){{\tilde \vartheta }_i} + {\tilde f}_i(t),
  \end{array}
\end{equation}
where ${{\tilde \vartheta }_i} = {{\hat \vartheta }_i} - \vartheta_i$ and
${\tilde f}_i(t) = {\hat H_i}({q_i}){\ddot{\tilde q}_{ri}} + {\hat C_i}({q_i},{{\dot q}_i}){\dot{\tilde q}_{ri}}
- J_i^T {\mathcal{K}_{xi}}J_i {\tilde s }_i - {\mathcal{K}_{si}} {\tilde s }_i
- {\mathcal{K}_{ri}} [{\rm sgn} ( \hat s_i ) - {\rm sgn} ( s_i )]$.
The boundedness of the states $q_i(t)$, $\dot q_i(t)$, $s_i(t)$, and
${\tilde f}_i(t)$, which will be used hereinafter, is analyzed in the following theorem.

\begin{theoremx}\label{t2}
Suppose that Assumptions A1 and A2 hold.
The distributed sliding-mode estimator (\ref{1.9})
guarantees that there exists a settle time $t_f \in (t_0,\infty)$
such that ${\tilde f}_i(t) = 0$ when $t \geq t_f$, $\forall i \in \mathcal{V}$.
Moreover, using the DCEA (\ref{1.2}) and (\ref{1.7}) for (\ref{1.1}),
the states $q_i(t)$, $\dot q_i(t)$, and $s_i(t)$ will remain bounded
for bounded initial values when $t \in {\mathcal Q}_f = [t_0,t_f)$.
\end{theoremx}

\begin{proof}
For the first presentation, we prove that there exists a finite time
$t_f \in (t_0,\infty)$ such that ${\tilde f}_i(t) = 0$ when $t \geq t_f$.
Because the right-hand side of (\ref{1.9}) is discontinuous,
the differential inclusions and nonsmooth analysis are invoked for convergence analysis of (\ref{1.9}) \cite{Fischer,Paden}.
The error dynamics of (\ref{1.9}) can be rewritten as
\begin{equation*}
  {\dot \sigma_i} \in^{a.e.} \mathbb{K} \left\{ - \left( \left[ {\begin{array}{*{20}{lll}}
                            {{\beta _1}}&{}&{}\\
                            {}&{{\beta _2}}&{}\\
                            {}&{}&{{\beta _3}}
                            \end{array}} \right] \otimes I_{m} \right)
    {\rm sgn} \left( {\sum\limits_{j \in {\mathcal N}_i} {{\varepsilon_{ij}}\sigma_{ij}} } + b_i \sigma_i \right)
    - \left[ {\begin{array}{*{20}{lll}}
                            \dot x_0 \\
                            \dot v_0 \\
                            \dot a_0
                            \end{array}} \right] \right\},
\end{equation*}
where $a.e.$ stands for ``almost everywhere'' refer to almost all $t \in {\mathcal Q}$,
and $\mathbb{K}\{ \cdot \}$ denotes the differential inclusion \cite{Paden}.
By Theorem 3.1 in {\cite{Cao01}}, under Assumptions A1 and A2,
it can be obtained that $\sigma_i = 0$ when $t \geq t_f$ and $t_f$ is bounded by
\begin{equation*}
\begin{array}{lll}
  t_f \leq t_{f\max} := \max (t_{f1},t_{f2},t_{f3}) < \infty, \\ \\
  t_{f1} = t_0 + \frac{\max_{i \in \mathcal V} \| {\hat x}_i(t_0) - x_0(t_0) \|_\infty}{\beta_1 - \sup_{t \in \mathcal Q}\| v_0 \|_{\infty}}, \\ \\
  t_{f2} = t_0 + \frac{\max_{i \in \mathcal V} \| {\hat v}_i(t_0) - v_0(t_0) \|_\infty}{\beta_2 - \sup_{t \in \mathcal Q}\| a_0 \|_{\infty}}, \\ \\
  t_{f3} = t_0 + \frac{\max_{i \in \mathcal V} \| {\hat a}_i(t_0) - a_0(t_0) \|_\infty}{\beta_3 - \sup_{t \in \mathcal Q}\| \dot a_0 \|_{\infty}},
  \end{array}
\end{equation*}
where $\max(\cdot)$ and $\sup(\cdot)$ denote the maximal value and the supremum respectively.
It follows from (\ref{1.0}) and (\ref{2.2}) that $\dot{\tilde q}_{ri} = 0$ and $\ddot{\tilde q}_{ri} = 0$
when $t \geq t_f$.
Therefore, ${\tilde f}_i(t) = 0$ when $t \geq t_f$, $\forall i \in \mathcal{V}$.
\par
For the second presentation, using the DCEA (\ref{1.2}) and (\ref{1.7}) for (\ref{1.1}),
it is shown that the states $q_i(t)$, $\dot q_i(t)$, and $s_i(t)$ will remain bounded
for bounded initial values when $t \in {\mathcal Q}_f$.
Note that
\begin{equation*}
  \left\| \left[ {\begin{array}{*{20}{lll}}
                            {{\beta _1}}&{}&{}\\
                            {}&{{\beta _2}}&{}\\
                            {}&{}&{{\beta _3}}
                            \end{array}} \right] \otimes I_{m} \right\|_\infty \leq \max(\beta _1,\beta _2,\beta _3),
\end{equation*}
and
\begin{equation*}
  \left\| {\rm sgn} \left( {\sum\limits_{j \in {\mathcal N}_i} {{\varepsilon_{ij}}\sigma_{ij}} } + b_i \sigma_i \right) \right\|_\infty \leq 1.
\end{equation*}
It thus follows from (\ref{1.9}) that
\begin{equation*}
  \left\| {\dot \zeta_i} \right\|_\infty \leq \max(\beta _1,\beta _2,\beta _3), ~\forall i \in \mathcal V,
\end{equation*}
which means
\begin{equation*}
  \mathop {\sup }\limits_{{t \in {\mathcal Q}_f}}  \left\| \zeta_i(t) \right\|_\infty \leq \max(\beta _1,\beta _2,\beta _3)(t_f - t_0) < \infty, ~\forall i \in \mathcal V.
\end{equation*}
It follows that ${{\hat x}_i}(t)$, ${\hat v}_i(t)$, and ${\hat a}_i(t)$ remain bounded
for bounded initial values ${{\hat x}_i}(t_0)$, ${\hat v}_i(t_0)$, and ${\hat a}_i(t_0)$ when $t \in {\mathcal Q}_f$.
For bounded $q_i$ and $\dot q_i$, the kinematics in (\ref{1.1}) implies that $x_i$ and $\dot x_i$ remain bounded.
Then (\ref{1.5}) and (\ref{1.6}) implies that ${\dot{\hat q}_{ri}}$, ${\ddot{\hat q}_{ri}}$
and $\hat s_i$ remain bounded for bounded $q_i$ and $\dot q_i$.
By (\ref{1.3})-(\ref{2.2}), ${\dot q_{ri}}$, ${\ddot q_{ri}}$, $s_i$, ${\dot{\tilde q}_{ri}}$, ${\ddot{\tilde q}_{ri}}$,
and $\tilde s_i$ remain bounded for bounded $q_i$ and $\dot q_i$.
By Property P2 and Remark 1, $Y_i( {q_i},{\dot q}_i,{\dot{\hat q}_{ri}},{\ddot{\hat q}_{ri}} )$ are bounded for bounded $q_i$ and $\dot q_i$.
Then (\ref{1.8}) implies ${{\hat \vartheta }_i}(t)$ remain bounded when $t \in {\mathcal Q}_f$
for bounded initial value ${{\hat \vartheta }_i}(t_0)$.
It follows that ${\tilde f}_i(t)$ remain bounded for bounded $q_i$ and $\dot q_i$.
By (\ref{2.3}), ${\dot s_i}$ remains bounded for bounded $q_i$ and $\dot q_i$.
Therefore, we can get that for bounded initial values $q_i(t_0)$ and $\dot q_i(t_0)$,
the states $q_i(t)$, $\dot q_i(t)$, and $s_i(t)$ remain bounded when $t \in {\mathcal Q}_f$.
This completes the proof.
\end{proof}

\subsection{Convergence Analysis}

In this section, based on Theorem \ref{t2}, the convergence of the close-loop dynamics is analyzed.
Differential inclusion and Filippov solution \cite{Paden} are invoked because the presented DCEA are nonsmooth.

By Theorem \ref{t2}, $\hat s_i = s_i$ and ${\tilde f}_i(t) = 0$ when $t \in \bar{\mathcal Q}_f = [t_f, \infty)$.
Then when $t \in \bar{\mathcal Q}_f$, the combination of (\ref{1.7}) and (\ref{2.3}) yields the following cascade system
\begin{equation}\label{2.5}
  \left\{ \begin{array}{lll}
  {\dot s_i} = {H_i^{-1}}({q_i}) [ - {C_i}({q_i},{{\dot q}_i}){s_i} - J_i^T {\mathcal{K}_{xi}}J_i {s_i}
  - {\mathcal{K}_{si}} {s_i} \\
  ~~~~~~~~~~~~~~~~~ + Y_i( {q_i},{\dot q}_i,{\dot q_{ri}},{\ddot q_{ri}} ){{\tilde \vartheta }_i} + {\tilde h}_i(t)], \\ \\
   {x_i} = {\psi _i}({q_i}),~{{\dot x}_i} = {J_i} {{\dot q}_i}, \\ \\
  {{\dot {\hat \vartheta} }_i} = - {\mathcal T}_i Y_i^T( {q_i},{\dot q}_i,{\dot q_{ri}},{\ddot q_{ri}} ){s_i},
  \end{array} \right.
\end{equation}
where $q_i(t_f)$, $\dot q_i(t_f)$, $x_i(t_f)$, $\dot x_i(t_f)$, and $s_i(t_f)$ are bounded,
${\tilde h}_i(t) = - {d_i}(t) - {\mathcal{K}_{ri}} {\rm sgn} ( s_i )$
denotes the combination of the nonsmooth and uncertain terms in the system dynamics.
Then the following theorem is given.

\begin{theoremx}\label{t1}
Suppose that Assumptions A1 and A2 hold.
Using the DCEA (\ref{1.2}) and (\ref{1.7}) for (\ref{1.1}),
if $\lambda_{\min}(\mathcal{K}_{xi}) > 0$, $\lambda_{\min}(\mathcal{K}_{si}) > 0$,
$\lambda_{\min}(\mathcal{K}_{ri}) \geq \sup_{t \in \mathcal Q} \| d_i(t) \|$,
and $\lambda_{\min}({\mathcal T}_i) > 0$,
then the control tasks in this paper can be achieved,
namely, ${x_i} \to {x_0}$ and ${\dot{x}_i} \to {v_0}$ as $t \to \infty$,
$\forall i \in {\mathcal V}$; $e_{si} \to 0$ as $t \to \infty$,
$\forall i \in {\mathcal F}$.
\end{theoremx}

\begin{proof}
The proof proceeds in three steps.
First, the passivity of the close-loop dynamics (\ref{2.5}) is analyzed.
Second, the convergence of $s_i$ and $J_i s_i$ is obtained
using differential inclusions and nonsmooth analysis.
Finally, the stability of $x_i$ and $e_{si}$ is shown
invoking kinematic analysis of nonredundant and redundant manipulators.

The first presentation presents the passivity of the close-loop dynamics.
For system (\ref{2.5}), consider the following storage function
\begin{equation*}
{V_{i1}} = \frac{1}{2}(s _i^T{H_i}({q_i}){s _i} + \tilde \vartheta _i^T\mathcal{T}_i^{ - 1}{\tilde \vartheta _i} ).
\end{equation*}
Differentiating the storage function along (\ref{2.5}) gives
\begin{equation*}
\begin{array}{lll}
{\dot V}_{i1} &=& \frac{1}{2}s _i^T{{\dot H}_i}({q_i}){s _i} + s _i^T{H_i}({q_i}){{\dot s }_i}
+ \tilde \vartheta _i^T\mathcal{T}_i^{ - 1}{{\dot {\hat \vartheta }}_i} \\ \\
&=&  - s _i^T \mathcal{K}_{si} {s _i} - s _i^T J_i^T(q_i) \mathcal{K}_{xi} J_i(q_i) {s }_i + s _i^T {\tilde h}_i \\ \\
&\leq& s _i^T {\tilde h}_i,
\end{array}
\end{equation*}
where Property P1 is invoked to obtain the second equation.
Integrating both sides of the above inequality with respect to time $t \in \bar{\mathcal Q}_f$ provides
\begin{equation}\label{4.0}
{V_{i1}}(t) - {V_{i1}}({t_f}) \leq \int_{{t_f}}^t {s _i^T( \omega){\tilde h}_i(\omega )d\omega }.
\end{equation}
It thus follows that system (\ref{2.5}) is passive with the
mapping from the input ${\tilde h}_i(t)$ to the state $s_i$, $\forall i \in \mathcal{V},~t \in \bar{\mathcal Q}_f$ \cite{Khalil}.
\par
The second presentation analyze the convergence of $s_i$ and $J_i s_i$ using nonsmooth analysis
and the passivity property given in the first presentation.
Let $\eta_i \in (0,\lambda_{\min}(\mathcal{K}_{xi}))$ be a positive scalar.
Then when $t \in \bar{\mathcal Q}_f$, consider the Lyapunov function candidate for system (\ref{2.5}) as
\begin{equation*}
{V_i} = {V_{i1}} + \eta_i \int_{t_f}^t s_i^T(\omega) J_i^T(q_i(\omega)) J_i(q_i(\omega)) {s_i}(\omega) d\omega.
\end{equation*}
Considering that ${\tilde h}_i(t)$ is nonsmooth and uncertain,
the generalized time derivative is invoked to carry out more formal mathematical analysis;
additionally, considering that the signum function is measurable
and locally essentially bounded, the Filippov solution exists for (\ref{2.5}) \cite{Fischer,Paden}.
Taking the generalized time derivative of $V_i$ along (\ref{2.5}) gives that
\begin{equation*}
\begin{array}{lll}
\dot {\tilde V}_i &=& \mathbb{K}\{ \frac{1}{2} s _i^T{{\dot H}_i}({q_i}){s _i} + s _i^T[ - {C_i}({q_i},{{\dot q}_i}){s_i} - J_i^T {\mathcal{K}_{xi}}J_i {s_i} - {\mathcal{K}_{si}} {s_i} \\
&& ~~~+ Y_i( {q_i},{\dot q}_i,{\dot q_{ri}},{\ddot q_{ri}} ){{\tilde \vartheta }_i} + {\tilde h}_i(t)] - \tilde \vartheta _i^T Y_i^T( {q_i},{\dot q}_i,{\dot q_{ri}},{\ddot q_{ri}} ){s_i} \\
&& ~~~+ \eta_i s_i^T J_i^T J_i {s_i} \} \\ \\
&=&  \mathbb{K}\{ - s _i^T{\mathcal{K}_{si}}{s _i} - s _i^T J_i^T({\mathcal{K}_{xi}} - \eta_i I_{m})J_i s_i - s _i^T {\mathcal{K}_{ri}} {\rm sgn} ( s_i ) - s _i^T d_i \} \\ \\
&\leq& - s _i^T{\mathcal{K}_{si}}{s _i} - s _i^T J_i^T({\mathcal{K}_{xi}} - \eta_i I_{m})J_i s_i - \| s_i \| ( \lambda_{\min}(\mathcal{K}_{ri}) - \| d_i \|) \\ \\
&\leq& - s _i^T{\mathcal{K}_{si}}{s _i} - s _i^T J_i^T({\mathcal{K}_{xi}} - \eta_i I_{m})J_i s_i
- \| s_i \| ( \lambda_{\min}(\mathcal{K}_{ri}) - \sup_{t \in \mathcal Q} \| d_i(t) \|) \\ \\
&\leq& - s _i^T{\mathcal{K}_{si}}{s _i} - s _i^T J_i^T({\mathcal{K}_{xi}} - \eta_i I_{m})J_i s_i,
\end{array}
\end{equation*}
where (\ref{4.0}), $\lambda_{\min}(\mathcal{K}_{ri}) \geq \sup_{t \in \mathcal Q} \| d_i(t) \|$,
Property P1-P3, and the fact that $\mathbb{K} \{ f \} \equiv \{ f \}$ if $f$ is continuous \cite{Paden}
are invoked to obtain the above inequalities.
Provided $\lambda_{\min}(\mathcal{K}_{xi}) > 0$ and $\lambda_{\min}(\mathcal{K}_{si}) > 0$,
$\dot {\tilde V}_i$ is negative definite, $\forall i \in \mathcal V$.
It follows that ${V}_i \in {\mathcal L}_\infty$;
then $s _i, J_i s_i, \tilde \vartheta _i \in {\mathcal L}_\infty$, $\forall i \in \mathcal V$.
$\tilde \vartheta _i \in {\mathcal L}_\infty$ implies that $\hat \vartheta _i \in {\mathcal L}_\infty$
because $\vartheta _i$ is a vector of constant dynamic parameters, $\forall i \in \mathcal V$.
$s _i, J_i s_i \in {\mathcal L}_\infty$ gives that
$Y_i( {q_i},{\dot q}_i,{\dot q_{ri}},{\ddot q_{ri}} ) \in {\mathcal L}_\infty$, $\forall i \in \mathcal V$.
Then (\ref{1.2}) and (\ref{1.7}) imply that the control input $u_i$ is bounded.
Thus, the closed-loop dynamics (\ref{2.5}) implies that $\dot s_i \in {\mathcal L}_\infty$;
then $s_i$ is uniformly continuous, $\forall i \in \mathcal V$.
Then invoking Corollary 2 in \cite{Paden}, $s _i \to 0$ and $J_i s_i \to 0$ as $t \to \infty$, $\forall i \in \mathcal V$.
\par
The third presentation shows that $e_{si} \to 0$ $(\forall i \in \mathcal F)$ as $t \to \infty$ under the DCEA.
Let $e_i = x_i - x_0$. (\ref{1.3}) gives that
\begin{equation}\label{2.8}
  {\dot e_i} =  - {\alpha _i}{e_i} + {J_i} {s _i}.
\end{equation}
It thus follows from \cite{Khalil} that (\ref{2.8}) is input-to-state stable
with respect to the input ${J_i}({q_i}){s _i}$ and the state $e_i$;
hence, $J_i s_i \to 0$ as $t \to \infty$ implies that $e_i \to 0$ as $t \to \infty$, which means $x_i \to x_0$ and $\dot x_i \to v_0$ as $t \to \infty$, $\forall i \in \mathcal V$, $i.e.$, the maintask is addressed.
Additionally, by (\ref{1.3}) and Lemma 1, for the redundant manipulators, namely, for $i \in \mathcal F$,
\begin{equation*}
    \begin{array}{lll}
    {e_{si}} &=& ( I_{{p_i}} - J_i^ \sharp {J_i} )( \dot q_i - {\varphi _i} ) \\ \\
             &=& ( I_{{p_i}} - J_i^ \sharp {J_i} )[\dot q_i - J_i^\sharp (v_0 - {\alpha _i} (x_i - x_0))
             - (I_{p_i} - J_i^\sharp {J_i}){\varphi _i} ] \\ \\
             &=& ( I_{{p_i}} - J_i^ \sharp {J_i} )(\dot q_i - \dot q_{ri}) \\ \\
             &=& ( I_{{p_i}} - J_i^ \sharp {J_i} )s_i.
    \end{array}
\end{equation*}
Therefore, $s_i \to 0$ as $t \to \infty$ implies that ${e_{si}} \to 0$ as $t \to \infty$,
$\forall i \in\mathcal F$, $i.e.$, the subtask is addressed.
This completes the proof.
\end{proof}

Note that the necessary and sufficient condition can be obtained by
some simple transformation for Theorem \ref{t1}.
The proof of Corollary \ref{c1} can be easily obtained by contradiction and is omitted here.

\begin{corollaryx}\label{c1}
Suppose that $\lambda_{\min}(\mathcal{K}_{xi}) > 0$, $\lambda_{\min}(\mathcal{K}_{si}) > 0$,
$\lambda_{\min}(\mathcal{K}_{ri}) \geq \sup_{t \in \mathcal Q} \| d_i(t) \|$, $\lambda_{\min}({\mathcal T}_i) > 0$,
and Assumption A2 holds. Using (\ref{1.2}) and (\ref{1.7}) for (\ref{1.1}),
the control tasks are achieved, namely, ${x_i} \to {x_0}$ and ${\dot{x}_i} \to {v_0}$ as $t \to \infty$,
$\forall i \in {\mathcal V}$; $e_{si} \to 0$ as $t \to \infty$, $\forall i \in {\mathcal F}$, if and only if Assumption A1 holds.
\end{corollaryx}

It is worthy to point out that the DCEA (\ref{1.2}) and (\ref{1.7}) can also deal with the consensus problem,
namely, leaderless coordination, presented in \cite{Guan01} by some simple changes.
\begin{corollaryx}\label{c2}
Suppose that $\mathcal G$ contains a spanning tree and Assumption {A2} holds. Let estimators (\ref{1.9}) be replaced by
\begin{equation*}
  {\dot \zeta_i} = - \left( \left[ {\begin{array}{*{20}{lll}}
                            {{\beta _1}}&{}&{}\\
                            {}&{{\beta _2}}&{}\\
                            {}&{}&{{\beta _3}}
                            \end{array}} \right] \otimes I_{m} \right)
    {\rm sgn} \left( {\sum\limits_{j \in {\mathcal N}_i} {{\varepsilon_{ij}}\sigma_{ij}} } \right).
\end{equation*}
Using (\ref{1.2}), (\ref{1.8}) and the above estimators for the MHMs (\ref{1.1}), if $\lambda_{\min}(\mathcal{K}_{xi}) > 0$, $\lambda_{\min}(\mathcal{K}_{si}) > 0$, $\lambda_{\min}(\mathcal{K}_{ri}) \geq \sup_{t \in \mathcal Q} \| d_i(t) \|$, and $\lambda_{\min}({\mathcal T}_i) > 0$, then consensus in \cite{Guan01} is achieved, meanwhile, the subtasks are obtained for $i \in {\mathcal F}$.
\end{corollaryx}

Let a switching topology $\mathcal{G}(t) =  \{\mathcal{V},{\xi}(t),\mathcal{A}(t) \}$ be the system interaction.
$\mathcal{A}(t) = \left[ { \varepsilon_{ij}}(t) \right]_{n \times n}$ is the weighted adjacency matrix.
$\mathcal{B}(t) = [{b_1(t)},\ldots,{b_n(t)}]^T$ is the time-varying
nonnegative weight vector between the $n$ nodes and the leader. We can obtain the following corollary.

\begin{corollaryx}\label{c3}
Suppose that Assumption A2 holds and the leader is reachable to the MHMs under $\mathcal{G}(t)$ and $\mathcal{B}(t)$.
Let (\ref{1.9}) be replaced by
\begin{equation*}
  {\dot \zeta_i} = - \left( \left[ {\begin{array}{*{20}{lll}}
                            {{\beta _1}}&{}&{}\\
                            {}&{{\beta _2}}&{}\\
                            {}&{}&{{\beta _3}}
                            \end{array}} \right] \otimes I_{m} \right)
    {\rm sgn} \left( {\sum\limits_{j \in {\mathcal N}_i} {{\varepsilon_{ij}(t)}\sigma_{ij}} } + {b_i(t)}\sigma_i \right).
\end{equation*}
Using (\ref{1.2}), (\ref{1.8}) and the above estimators for the MHMs (\ref{1.1}), if $\lambda_{\min}(\mathcal{K}_{xi}) > 0$, $\lambda_{\min}(\mathcal{K}_{si}) > 0$, $\lambda_{\min}(\mathcal{K}_{ri}) \geq \sup_{t \in \mathcal Q} \| d_i(t) \|$, and $\lambda_{\min}({\mathcal T}_i) > 0$,
then ${x_i} \to {x_0}$ and ${\dot{x}_i} \to {v_0}$ as $t \to \infty$,
$\forall i \in {\mathcal V}$; $e_{si} \to 0$ as $t \to \infty$, $\forall i \in {\mathcal F}$.
\end{corollaryx}
\begin{proof}
The proof can be easily developed by the combination of Lemma 6 in \cite{Wang03}
and Theorem \ref{t1}, and is omitted here.
\end{proof}

Assuming that $\mathcal G$ is balanced connected or detail balanced,
other conditions are identical, the above results still hold,
meanwhile, the subtasks are obtained for all $i \in {\mathcal F}$.

\begin{remark}
The DCEA (\ref{1.2}) and (\ref{1.7}) provide a distributed control framework for various multi-agent networks.
Note that (\ref{1.2}) and $(\ref{1.8})$ deal with the tracking of the estimated value
of the $i$th manipulator using its states (local states), and are called local algorithms,
while (\ref{1.9}) deals with the estimation of the states of the leader using both
the states of the $i$th manipulator and its neighbour set, and is called networked algorithms.
Thus, the DCEA consists of local algorithms and networked algorithms,
which can be easily applied to various nonlinear multi-agent networks by designing local algorithms and networked algorithms separately.
\end{remark}

\begin{remark}
Note that Assumption A1 is less conservative than many other assumptions in recent works.
For details, let $\bar {\mathcal G}$ be the augmentation digraph by adding the leader and the weight vector $\mathcal B$ to $\mathcal{G}$, Assumption A1 means that $\bar {\mathcal G}$ contains a spanning tree, which is less conservative than the following assumptions for graphs, including strong connected {\cite{Liu02,Wang01,Wang02}}, balanced connected {\cite{Liu01}}, detail balanced {\cite{LiuLam}},
which means, less consumption is required to establish and maintain Assumption A1.
\end{remark}

\begin{remark}
Assumption A1 implies that the presented DCEA have the following advantages.
(1) If ${\mathcal{G}}$ contains one or more edges, then only partial nodes are required to be pinned,
which reflects the superiority of networked control.
(2) By weighing the cost of maintaining interaction and pinning nodes, the interaction topology can be optimized.
\end{remark}

\begin{remark}
It is worthy to point out that the leader considered in Theorem \ref{t1} is a bounded time-varying signal,
which is necessary and inevitable in some natural and man-made systems,
$e.g.$, humanoid hands and multi-fingered hands {\cite{fmg03,Gueaieb01}}.
\end{remark}

\section{Simulations}

In this section, simulations are performed to show the effectiveness of the presented DCEA.
The MHMs contain five two-DOF and two three-DOF planar manipulators, as shown in Fig.\ref{Fig.Mod},
and the physical parameters are given in Tab.\ref{tab.0}.
The dynamics and kinematics of the MHMs are adopted from \cite{Wang01}.
The input disturbance $d_i(t)$ is a stochastic signal, whose $\infty$-norm is bounded by 40.
The maintask is to obtain the tracking of the leader
\begin{equation}\label{0.5}
  {x_0}(t) = \left[ {\begin{array}{*{20}{lll}}
                    {1.2 + 0.5\sin (\pi t)}\\
                    {1.3 + 0.3\cos (\pi t)}
                    \end{array}} \right]
\end{equation}
for the end-effectors in the \emph{XY} plane.
The MHMs are interacted invoking the digraph $\mathcal G$ shown in Fig.\ref{Fig.Com},
where node 0 is the leader, nodes 1-5 are two-DOF nonredundant manipulators, nodes 6 and 7 are three-DOF redundant manipulators,
and nodes 1, 3, 5, and 7 can access the information of the leader directly.
For simplification, if node $i$ can access the information of node $j$ directly,
$\varepsilon_{ij} = 1$; otherwise $\varepsilon_{ij} = 0$, $\forall i,j = 1,\ldots,7$.
${\mathcal B} = {\rm col} (1,0,1,0,1,0,1)$. The Laplacian matrix with respect to the digraph $\mathcal G$ is given by
\begin{equation*}
  {\mathcal L} =
  \left[ {\begin{array}{*{20}{c}}
{\begin{array}{*{20}{c}}
\begin{array}{l}
0\\
 - 1\\
0\\
 - 1\\
0\\
0\\
0
\end{array}&\begin{array}{l}
0\\
2\\
0\\
0\\
0\\
0\\
0
\end{array}&\begin{array}{l}
0\\
 - 1\\
0\\
 - 1\\
0\\
0\\
0
\end{array}&\begin{array}{l}
0\\
0\\
0\\
2\\
0\\
0\\
0
\end{array}
\end{array}}&\begin{array}{l}
0\\
0\\
0\\
0\\
0\\
 - 1\\
0
\end{array}&\begin{array}{l}
0\\
0\\
0\\
0\\
0\\
1\\
0
\end{array}&\begin{array}{l}
0\\
0\\
0\\
0\\
0\\
0\\
0
\end{array}
\end{array}} \right].
\end{equation*}
\par
The elements of $q_i(0)$, $\dot q_i(0)$, and $\zeta_i(0)$ are randomly chosen from $[-5,5]$.
The elements of ${\hat \vartheta}_i(0)$ are randomly selected from $[0,5]$.
The control parameters for the DCEA (\ref{1.2}) and (\ref{1.7}) are given by
\begin{equation*}
\begin{array}{lll}
  \beta_1 = 4,~\beta_2 = 7,~\beta_3 = 21; \\
  {\alpha _i} = 3,~\mathcal{T}_i = 0.1 I,~{\mathcal{K}_{xi}} = diag \{ 50,50 \},~\forall i = 1,\ldots,7; \\
  {\mathcal{K}_{si}} = diag \{ 100,100 \},~{\mathcal{K}_{ri}} = diag \{ 60,60 \},~\forall i = 1,\ldots,5; \\
  {\mathcal{K}_{si}} = diag \{ 150,150,150 \},~{\mathcal{K}_{ri}} = diag \{ 60,60,60 \},~\forall i = 6,7,
  \end{array}
\end{equation*}
where $I$ denotes the identity matrices of appropriate orders. The sampling period is selected as 10 ms.

\begin{table}[h]
 \centering\small
  \begin{tabular}{cccccc}
  \hline
  \hline
   $i$-th node & $m_i$~($kg$) & $l_i$~($m$) & $r_i$~($m$) & $I_i$~($kg\cdot m^2$) \\
  \hline
   $i = 1$ & 0.8, 0.6 & 1.4, 0.9　& 0.8, 0.45 & 6, 3 \\
   $i = 2$ & 1, 0.8 & 1.2, 1.1　& 0.7, 0.5 & 2, 3 \\
   $i = 3$ & 0.5, 0.8 & 1.1, 1.3　& 0.4, 0.6 & 5, 3 \\
   $i = 4$ & 1.5, 0.8 & 1.1, 1.2　& 0.6, 0.6 & 5, 4 \\
   $i = 5$ & 2.3, 0.8 & 1.0, 1.2　& 0.4, 0.7 & 5, 3 \\
   $i = 6$ & 0.8, 1.2, 1.4 & 0.8, 1.1, 1.4　& 0.4, 0.5, 0.7 & 4, 6, 5 \\
   $i = 7$ & 1.8, 1.2, 1.4 & 1, 1.1, 1.2　& 0.6, 0.6, 0.6 & 5, 6, 5 \\
  \hline
 \end{tabular}
 \caption{The physical parameters of the MHMs. }\label{tab.0}
\end{table}

\begin{figure}[H]
  \centering
  \includegraphics[width=9cm]{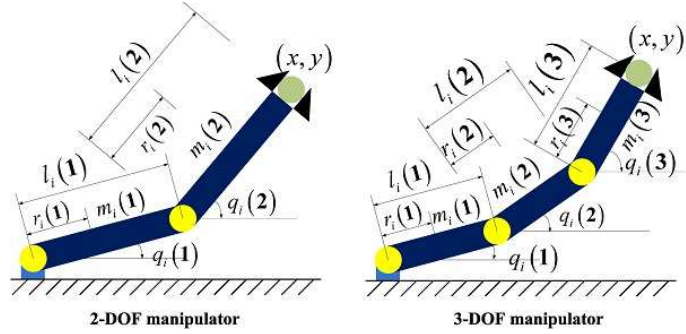}\\
  \caption{The mechanical structure of the individuals in the MHMs.}\label{Fig.Mod}
\end{figure}

\begin{figure}[H]
  \centering
  \includegraphics[width=6cm]{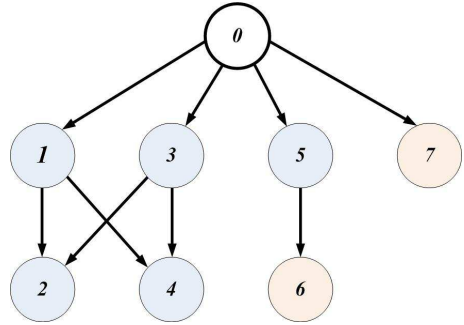}\\
  \caption{The interaction of MHMs $\mathcal G$ .}\label{Fig.Com}
\end{figure}

Simulation results of the presented DCEA are shown in Fig.\ref{Fig.1}-\ref{Fig.4}.
(\ref{0.5}) implies that $\| v_0 \|_\infty \leq 1.571$, $\| a_0 \|_\infty \leq 4.935$, and $\| \dot a_0 \|_\infty \leq 15.504$.
It thus follows from Fig.\ref{Fig.Com} that Assumptions A1 and A2 holds in this simulation.
Then the estimated states ${\hat x}_i$, ${\hat v}_i$ and ${\hat a}_i$ follows the states of the leader in finite time,
as shown in Fig.\ref{Fig.1}.
It gives that $\tilde f_i (t)$ converges to the origin in finite time and Theorem \ref{t2} holds, which means
the distributed sliding-mode estimator (\ref{1.9}) can drive the system dynamics into the cascade form (\ref{2.5}).
Moreover, the end-effectors of the manipulators follows the leader asymptotically because
$\lambda_{\min}(\mathcal{K}_{xi}) > 0$, $\lambda_{\min}(\mathcal{K}_{si}) > 0$
and $\lambda_{\min}(\mathcal{K}_{ri}) \geq \sup_{t \in \mathcal Q} \| d_i(t) \|$,
as shown in Fig.\ref{Fig.2} (the trajectories of the states in time-domain) and Fig.\ref{Fig.3}
(the trajectories of the states in \emph{XY} plane).
It means that the maintask can be accomplished under the DCEA in the simulations.
The subtask function for the sixth manipulator (node 6) is selected as
$\varphi_6 = {\rm col} [0,9 (1 - q_2),0]$ with respect to the performance indices $4.5 (q_2 - 1)^2$,
where $q_2$ denotes the second joint angle of the sixth manipulator.
This subtask function forces the second joint of the sixth manipulator toward 1 $rad$.
For Fig.\ref{Fig.5}, in the left picture, the second joint of the sixth manipulator is forced toward 1 $rad$
with subtask control; in the right picture, the subtask cannot be accomplished without subtask control.
The subtask function for the seventh manipulator is $\varphi_7 = \frac{\partial }{{\partial q}}\left( {\det ({J_7}J_7^T)} \right)$,
where $q$ denotes the joint position of the seventh manipulator.
This subtask function increases the manipulability of the seventh manipulator, as shown in \cite{Hsu}.
It can be observed in Fig.\ref{Fig.4} that the manipulability of the seventh manipulator is enhanced with subtask control.
It means that the subtask can also be accomplished under the DCEA in the simulations.

\begin{figure}[H]
  \centering
  \includegraphics[width=14cm]{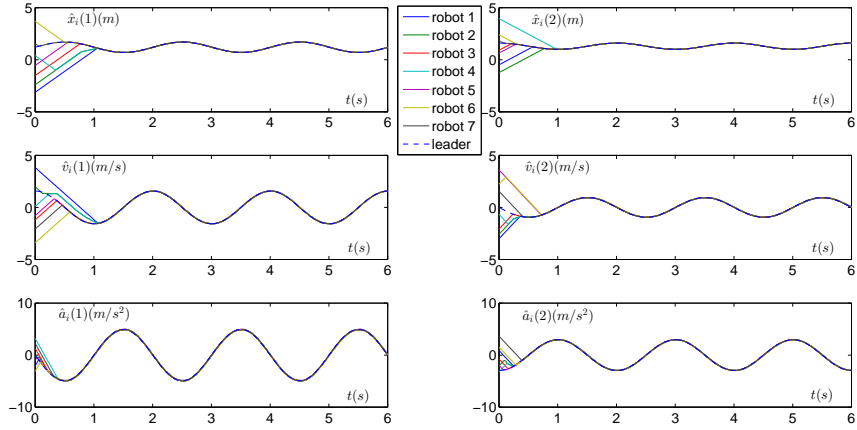}\\
  \caption{The estimated value $\hat x_i$, $\hat v_i$ and $\hat a_i$ for the $i$th manipulator.}\label{Fig.1}
\end{figure}

\begin{figure}[H]
  \centering
  \includegraphics[width=14cm]{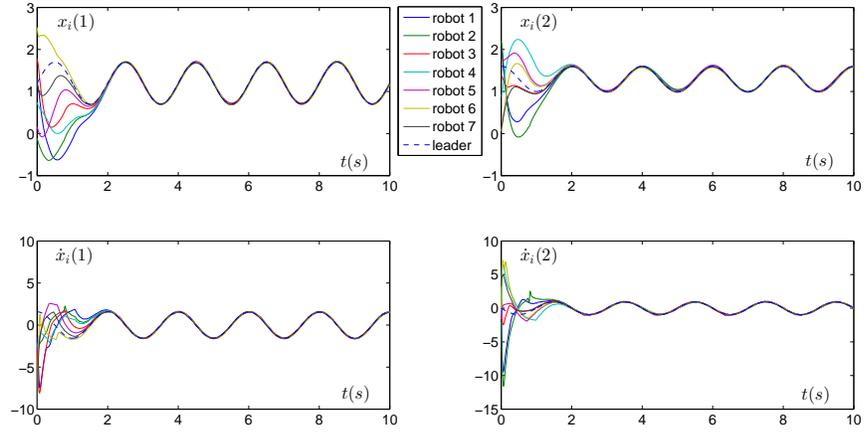}\\
  \caption{The task-space states $x_i$ and $\dot x_i$ for the $i$th manipulator.}\label{Fig.2}
\end{figure}

\begin{figure}[H]
  \centering
  \includegraphics[width=14cm]{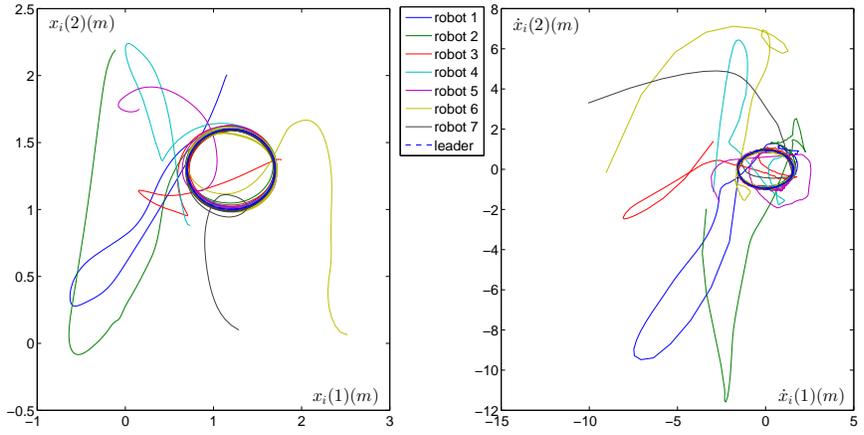}\\
  \caption{Position (the left picture) and velocity (the right picture) of the end-effectors of MHMs in the $XY$ plane.}\label{Fig.3}
\end{figure}

\begin{figure}[H]
  \centering
  \includegraphics[width=14cm]{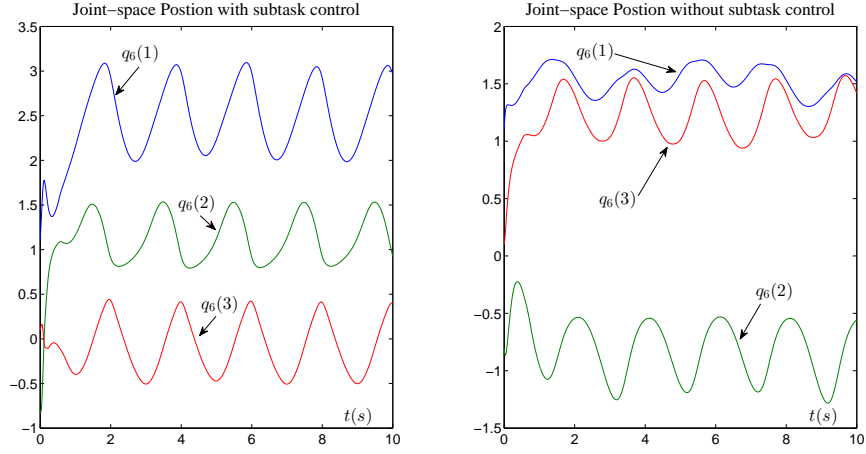}\\
  \caption{Joint-space position of the sixth manipulator with and without subtask control.}\label{Fig.5}
\end{figure}

\begin{figure}[H]
  \centering
  \includegraphics[width=14cm]{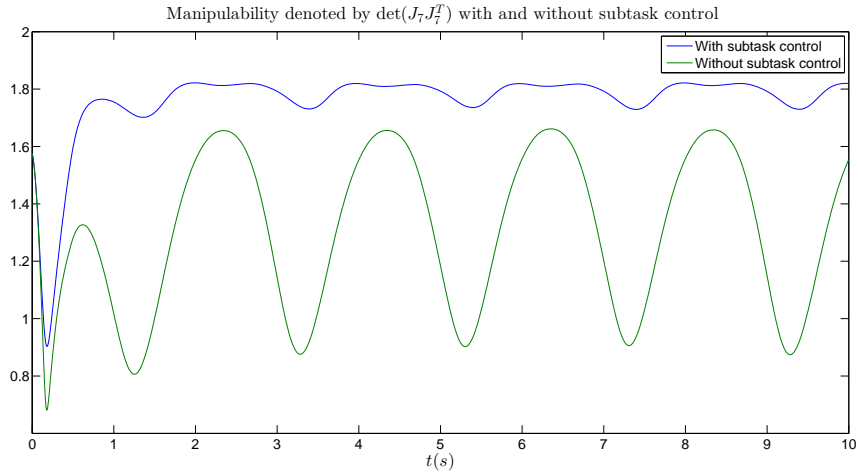}\\
  \caption{Manipulability of the seventh manipulator with and without subtask control.}\label{Fig.4}
\end{figure}

\section{Conclusion}
This paper focus on the task-space coordinated tracking problem of MHMs
with parametric uncertainties and input disturbances under digraphs.
Especially, maintask and subtask for MHMs are designed and addressed simultaneously.
Several conditions (including necessary and sufficient conditions, sufficient conditions) for driving
both the task-space tracking error and the subtask tracking error of MHMs to zero have been derived.
The simulations for the DCEA have shown satisfactory performance in the MHMs containing two-DOF and three-DOF manipulators.
\medskip

\section*{Acknowledgements}
This work was supported in part by the National Natural Science Foundation of China under Grant 61272069.


\bibliographystyle{elsarticle-num}
\bibliography{<your-bib-database>}

\end{document}